\documentclass[sigconf, nonacm]{aamas}

\usepackage{balance} 
\usepackage{xspace}
\usepackage{makecell}
\usepackage{siunitx}
\usepackage{pgfplots}
\usepackage[nameinlink]{cleveref}
\usepackage{tabularx}
\usepackage{multirow}
\usepackage{adjustbox}
\newcolumntype{C}{>{\centering\arraybackslash}X}
\newcolumntype{L}{>{\raggedright\arraybackslash}X}
\newcolumntype{R}{>{\raggedleft\arraybackslash}X}

\usepackage[aboveskip=-2em, belowskip=0.5em]{subcaption}
\usepackage{caption}
\usepackage{pgfplots}
\usepackage{mathtools}
\usepackage{fnpct}

\usepackage{thm-restate}

\makeatletter
\g@addto@macro\bfseries{\boldmath}
\makeatother

\DeclareUnicodeCharacter{2212}{−}
\usepgfplotslibrary{groupplots,dateplot}
\usetikzlibrary{patterns,shapes.arrows}
\pgfplotsset{compat=newest}

\title[Synergistic Traffic Assignment]{Synergistic Traffic Assignment}

\author{Thomas Bläsius}
\affiliation{
  \institution{Karlsruhe Institute of Technology}
  \city{Karlsruhe}
  \country{Germany}}
\email{thomas.blaesius@kit.edu}

\author{Adrian Feilhauer}
\affiliation{
  \institution{Karlsruhe Institute of Technology}
  \city{Karlsruhe}
  \country{Germany}}
\email{adrian.feilhauer@kit.edu}

\author{Markus Jung}
\affiliation{
  \institution{Karlsruhe Institute of Technology}
  \city{Karlsruhe}
  \country{Germany}}
\email{markus.jung@student.kit.edu}

\author{Moritz Laupichler}
\affiliation{
  \institution{Karlsruhe Institute of Technology}
  \city{Karlsruhe}
  \country{Germany}}
\email{moritz.laupichler@kit.edu}

\author{Peter Sanders}
\affiliation{
  \institution{Karlsruhe Institute of Technology}
  \city{Karlsruhe}
  \country{Germany}}
\email{sanders@kit.edu}

\author{Michael Zündorf}
\affiliation{
  \institution{Karlsruhe Institute of Technology}
  \city{Karlsruhe}
  \country{Germany}}
\email{michael.zuendorf@alumni.kit.edu}


\begin{abstract}
Traffic assignment analyzes traffic flows in road networks that emerge due to traveler interaction.
Traditionally, travelers are assumed to use private cars, so road costs grow with the number of users due to congestion.
However, in sustainable transit systems, travelers share vehicles s.t. more users on a road lead to higher sharing potential and reduced cost per user.
Thus, we invert the usual \emph{avoidant traffic assignment} (ATA) and instead consider \emph{synergistic traffic assignment} (STA) where road costs \emph{decrease} with use.

We find that STA is significantly different from ATA from a game-theoretical point of view.
We show that a simple iterative best-response method with simultaneous updates converges to an equilibrium state.
This enables efficient computation of equilibria using optimized speedup techniques for shortest-path queries.
In contrast, ATA requires slower sequential updates or more complicated iteration schemes that only approximate an equilibrium.
Experiments with a realistic scenario for the city of Stuttgart indicate that STA indeed quickly converges to an equilibrium.
We envision STA as a part of \emph{software-defined transportation} systems that dynamically adapt to current travel demand.
As a first demonstration, we show that an STA equilibrium can be used to incorporate traveler synergism in a simple bus line planning algorithm to potentially greatly reduce the required vehicle resources.  
\end{abstract}


\newcommand{\parheader}[1]{\paragraph*{#1}}

\DeclareMathOperator{\cost}{cost}
\DeclareMathOperator{\distance}{dist}

\newcommand{\lineLength}[1]{\ensuremath{\delta_{#1}}}
\newcommand{\lineVehTT}[1]{\ensuremath{T_{#1}}}

\makeatletter
\newcommand{\oset}[2]{%
  {\mathop{#2}\limits^{\vbox to 0.5\ex@{\kern-\tw@\ex@
   \hbox{\small #1}\vss}}}}
\makeatother

\begin{document}

\pagestyle{fancy}


\maketitle 


\section{Introduction}
Passenger transportation is in a crisis. 
Established systems often predominantly rely on the use of private cars, which leads to congestion, pollution, and bad usage of space in urban areas.
Traditional, statically planned public transportation is expensive, slow and inflexible. 
Thus, there is a lot of interest in more flexible systems that adapt to actual demand using a carefully orchestrated combination of individual transportation, shared taxis, and larger shared vehicles.
We are interested in evaluating the potential of such \emph{software-defined transportation} systems.

To this end, previous work evaluated the usage of large fleets of shared taxis in urban areas~\cite{buchhold2021fast,alonsomora2017ondemand,ho2018survey}. 
This alone has limited sharing potential which does not suffice to service the traveler demand in entire metropolitan areas.
Using walking/bicycles/scooters to reach pickup and dropoff points improves upon shared taxis but the overall savings are moderate (about 20\% less energy and slightly improved travel times)~\cite{laupichler2024fast,fielbaum2021ondemand}.
Thus, an essential component of efficient software-defined transportation is a way to identify roads with high sharing potential solely based on the (previously measured or simulated) demands of travelers (\emph{agents}). 
We model this using a \emph{traffic assignment} congestion game where agents' strategies are paths in a road network (a graph), resulting in loads on road segments (edges -- the resources of the game). 
We are interested in \emph{equilibria}, i.e., strategies (path choices) for each agent such that no agent can reduce their cost by unilaterally changing their path.

Traffic assignment has been intensively studied but so far exclusively for the \emph{avoidant} case (ATA) where costs increase with load due to congestion.
This is important when considering individual transportation.
For our application, load acts as a measure of sharing potential, which means costs \emph{decrease} with load.
We call this \emph{synergistic traffic assignment} (STA).
The focus of this paper is to analyze best response for STA, where agents react to the current situation with a strategy that is optimal for themselves.
We are interested in whether this converges to an equilibrium, with the ultimate goal to efficiently compute such a stable state.
Additionally, in \Cref{a:system_optima} we briefly discuss properties of system optima for STA and show that finding them is NP-hard, making them unsuitable for quick computation.

Motivated by the fact that shortest-path computations can be substantially accelerated if the edge costs do not change too often, we consider different best-response variants for STA.
In \emph{sequential} best response, each step consists of only one agent changing their strategy.
In \emph{simultaneous} best response, all agents choose their new strategy for the next iteration at the same time.
As an additional variation orthogonal to this, we distinguish between \emph{impact-aware} or \emph{impact-blind} best response.
In the impact-aware variant (typically considered in game theory), each agent takes their own impact on the cost function into account when considering to change their strategy.
In the impact-blind variant, agents make their decision on the current cost of edges, not on the cost after their own change.
We note that this makes little difference in practical scenarios with many agents where the impact of each individual agent on the cost is negligible.  
This yields four variants for STA for all combinations of sequential vs.\ simultaneous and impact-aware vs.\ impact blind; see \Cref{s:model} for formal definitions and \Cref{tab:variants_matrix} for an overview.

For the impact-aware case, it follows from the more general setting of congestion games (CG) that sequential best response converges to an equilibrium~\cite{rosenthal1973class}.
For simultaneous best response, however, we observe that there can be best-response cycles.
We additionally give an example of a best-response cycle in the case where a group of only two equivalent agents change their strategy simultaneously; see \Cref{sec:best-response-cycles}.
For the impact-blind case, we show that best response converges to an equilibrium, even in the simultaneous setting; see \Cref{sec:conv-best-resp}.
This difference between impact-blind and impact-aware is somewhat surprising, as they seem similar.
Moreover, it is in stark contrast to ATA, where impact-blindness may cause best-response cycles even in the sequential setting.

Beyond the theoretical result, the convergence of the simultaneous and impact-blind variant allows fast computation of an equilibrium.
In each round, all agents change simultaneously and base their choice on the same edge costs.
Thus, we can use sophisticated preprocessing techniques for path planning that are several orders of magnitude faster than individual applications of Dijkstra's algorithm; see \Cref{s:staAlgo}.
In contrast, ATA has to use expensive agent-by-agent updates or more complex iteration schemes that can use simultaneous updates but only approximate an equilibrium.

In \Cref{sec:experimental_evaluation}, we perform an experimental evaluation using realistic demand data for Stuttgart, Germany. 
Fewer than 20 iterations suffice to find an equilibrium. 
In contrast, established ATA algorithms are considered far from an equilibrium at this number of iterations~\cite{buchhold2019real,perederieieva2015framework}. 
Moreover, the employed shortest-path speedup technique of \emph{customizable contraction hierarchies} benefits from the fact that STA reinforces a hierarchy of roads, leading to faster queries in later iterations of STA.
Using the right parametrization of the cost function, we find routes that have only moderately (less than $25\%$) longer travel times than free-flow shortest paths while attaining considerably (more than two times) higher sharing potential.
We also look at a simple model for deriving bus lines from traffic flows.
We find that basing the bus lines on an STA equilibrium rather than free-flow traffic reduces the total vehicle operation time significantly.
Though the model is too simplistic to draw solid conclusions for real-world settings, it shows that STA can be a useful tool for such applications.

\parheader{Summary of Contributions}
\begin{itemize}
\item Introduction of synergistic traffic assignment (STA) as a way to extract sharing potential for
  a set of traveler demands
\item Theoretical analysis of the best-response process for different variants of STA
\item Practically efficient algorithm for finding an STA equilibrium based on the theoretical insights
\item Prototypical demonstration that STA equilibria identify high sharing potential in realistic inputs with uses in line planning
\end{itemize}
    
    \section{Related Work}
    To the best of our knowledge, there is no prior work on synergistic traffic assignment, i.e., traffic assignment with decreasing edge cost functions.
    Thus, here we consider research on avoidant traffic assignment (ATA) that may also be relevant for synergistic traffic assignment.
    ATA is a well-studied problem with decades of research in the operations and traffic planning communities.
    Additionally, ATA is often considered in game theory, as a prototypical example of a non-cooperative congestion game in pure strategies~\cite{Owen2013,rosenthal1973class}.

    \parheader{User Equilibrium}
    \citet{wardrop1952some} states that travelers in a road network naturally act selfishly, changing their path if a different path with a better travel time for themselves exists, even to the detriment of others.
    An assignment in which no traveler can change their path to their benefit is called a \emph{user equilibrium (UE)}.
    \citet{beckmann1956studies} find that a UE always exists for ATA, which matches a result from game theory about the existence of pure strategy Nash equilibria in congestion games~\cite{rosenthal1973class}.
    As traffic without central control naturally tends to a UE, it is an important issue of traffic analysis to compute UEs for a given network and traveler demand.

    \parheader{Simple Approaches to Finding a UE}
    A simple approach to finding a UE is a process called iterated \emph{sequential best response}, in which one traveler at a time may change their path while other travelers' paths are kept fixed.
    Edge costs are updated according to the change of one traveler before continuing with the next traveler.
    Sequential best response always converges to a UE~\cite{rosenthal1973class,Monderer1996}.
    However, the number of iterations until convergence can be exponential in the number of travelers~\cite{Durand2016, Fabrikant2004}.

    In iterated \emph{simultaneous best response}, each traveler chooses their next path simultaneously, while edge costs are kept fixed.
    Then, all paths are changed at once, before updating edge costs and repeating the process.
    While the simultaneous method needs fewer edge cost updates than sequential best response, it does not always converge.
    Instead, it may encounter best-response cycles where travelers revert to a previous state in a cyclical manner~\cite{sheffi1985urban}.
    Thus, simultaneous best response is also not suited to find a UE for ATA.

    \parheader{Finding a UE in Practice}
    In practice, more sophisticated methods find a UE by applying convex optimization to \emph{Beckmann's transformation}~\cite{beckmann1956studies} of ATA into a convex program whose minimum is a UE.
    A usual categorization~\cite{florian1995chapter,zhou2010computational,perederieieva2015framework} divides these algorithms into \emph{link-based}~\cite{frank1956algorithm,sheffi1985urban}, \emph{path-based}~\cite{dafermos1968traffic,jayakrishnan1994faster,florian2011new,kumar2011improved}, and \emph{bush-based}~\cite{bargera2002origin,dial2006path,nie2010class} approaches.
    These algorithms perform simultaneous path updates but they utilize the convexity of the problem to ensure progress towards the UE and avoid best-response cycles. 
    Note two limitations common to these approaches:
    First, travelers are treated as non-atomic, i.e., the load of a single O-D pair can be split among multiple paths.
    Second, the algorithms iteratively approach the UE in increasingly smaller steps but do not actually reach it.
    
    Each algorithm mentioned here requires an edge cost update and a re-computation of shortest paths in every iteration.
    The overhead for these updates can be expected to dominate the total running time. 
    \citet{buchhold2019real} show that a recent shortest-path speedup techniques called \emph{customizable contraction hierarchies (CCH)} can adequately reduce running times for cost updates and shortest-path queries to tenths of seconds per iteration for millions of O-D pairs. 

\section{Models and Notation}\label{s:model}
We define traffic assignment as a congestion game.
Let $[k] = \{1, \dots, k\}$ be a set of $k$ agents.
In a \emph{congestion game} there is a finite set $E$ of available \emph{resources} that are used by the agents.
Each agent $i \in [k]$ has a set of \emph{possible strategies} $P_i$ where each strategy $p \in P_i$ is a subset of resources, i.e., $p \subseteq E$.
A \emph{strategy profile} is a vector $S = (p_1, \dots, p_k)$ of strategies such that $p_i \in P_i$ for $i \in [k]$.
The \emph{load} $\ell_e(S)$ of a resource $e \in E$ with respect to a strategy profile $S$ is the number of agents using $e$, i.e., $\ell_e(S) = | \{p_i \in S \mid e \in p_i\}|$.
Each resource $e \in E$ has a \emph{cost} function $c_e \colon \mathbb N \to \mathbb R$, where $c_e(\ell)$ is the cost of resource $e$ given that it has load $\ell$.
The resulting cost for an agent is the sum of the costs of all resources used by the agent.
More formally, we define the \emph{cost} of a strategy $p \subseteq E$ given a strategy vector $S$ as $\cost(p, S) = \sum_{e \in p} c_e(\ell_e(S))$.
Thus, for $S = (p_1, \dots, p_k)$, the cost for agent $i \in [k]$ is $\cost(p_i, S)$.

With this, one can use the common definitions for the Nash equilibrium and best response.
Before defining them formally, we introduce the traffic assignment game, which also motivates less common variants of these concepts.
Colloquially speaking, in traffic assignment, each agent chooses a path in a graph that gets them from their origin to their destination, using the edges as resources.

More formally, let $G = (V, E)$ be a directed graph with a cost function $c_e \colon \mathbb N \to \mathbb R$ for each edge.
An \emph{origin--destination pair (O-D pair)} is a pair $(s, t) \in V \times V$ with $s \neq t$.
For $s, t \in V$, an \emph{$st$-path} is a set of edges $p \subseteq E$ such that the subgraph induced by $p$ contains $s$ and $t$, all vertices except $s$ have exactly one incoming edge, and all vertices except $t$ have exactly one outgoing edge.
For a set of O-D pairs $X = \{(s_1, t_1), \dots, (s_k, t_k)\}$, the \emph{traffic assignment game}, is a congestion game where the edge set $E$ forms the resources and the possible strategies for agent $i \in [k]$ are the $s_it_i$-paths, i.e., $P_i = \{p \subseteq E \mid p \text{ is an }s_it_i\text{-path}\}$.
A traffic assignment game is \emph{avoidant} if the cost functions $c_e$ are non-decreasing, i.e., if more traffic only increases cost.
Conversely, it is \emph{synergistic} if the cost functions $c_e$ are non-increasing.
We also use the terms avoidant and synergistic for the more general case of congestion games.

\subsection{Best Response}

Consider a congestion game and let $S = (p_1, \dots, p_k)$ be a strategy profile.
For an agent $i \in [k]$, we use $S_{-i}$ to denote the omission of $i$'s strategy.
Moreover, for any alternative strategy $p_i' \in P_i$ of agent $i$, we use $(S_{-i}, p_i')$ to denote the replacement of $p_i$ with $p_i'$ in $S$, i.e., $(S_{-i}, p_i') = (p_1, \dots, p_{i - 1}, p_i', p_{i + 1}, \dots, p_k)$.
Note that $S = (S_{-i}, p_i)$.

A \emph{best response} of agent $i$ to a strategy profile $S = (p_1, \dots, p_k)$ is a strategy $p_i' \in P_i$ that minimizes $\cost(p_i', (S_{-i}, p_i'))$, which is the cost for agent $i$ after they change their strategy from $p_i$ to $p_i'$.
We say that agent $i$ is \emph{content} with $S$ if $p_i$ is a best response to $S$.
If all agents are content with $S$, then $S$ is a \emph{Nash equilibrium}.
With \emph{sequential best response}, we refer to the process of changing the strategy of one agent at a time to their best response with respect to the current strategy profile.
If this reaches an equilibrium, we say that the process \emph{converges}.
Otherwise, we obtain a so-called \emph{best-response cycle}, in which the process indefinitely cycles through a repeating sequence of strategy profiles (these are the only two options as the state space is finite).
The following theorem by Rosenthal~\cite{rosenthal1973class} states that the sequential best response defined above converges.

\begin{theorem}[Rosenthal~\cite{rosenthal1973class}]
  \label{thm:rosenthal-congestion-games}
  In a congestion game, sequential best response converges to a Nash equilibrium.
\end{theorem}

As traffic assignment games are a special case of congestion games, \Cref{thm:rosenthal-congestion-games} also holds for our setting.
This is algorithmically useful as it provides a way to compute an equilibrium.
In each step, one has to compute a shortest path in $G$ for the O-D pair of the currently considered agent.
Unfortunately, this has the effect that the costs of the edges change after every shortest-path computation, obstructing the use of algorithmic shortest-path techniques that are based on pre-computations.
We overcome this by incorporating two types of adjustments to the above best-response variant.

\subsection{Simultaneous Best Response}

\emph{Simultaneous best response} is a process where, in every step, all agents change their strategy to their best response with respect to the current strategy profile.
While this feels like a step in the right direction in the sense that the strategy profile changes less frequently in relation to the number of shortest-path computations, there are two major downsides.
First, we show in~\Cref{sec:best-response-cycles} that this leads to a best-response cycle.
This is not very surprising as, in the synergistic setting, two agents may attempt to join each other and thereby avoid each other due to simultaneously changing their strategy.
We extend this by also providing a best-response cycle in a setting that is almost sequential in the sense that only two equivalent agents choose their new path together; see~\Cref{sec:best-response-cycles}.

The second downside is that the simultaneous setting alone is not enough to let us use pre-computation techniques for the shortest-path queries.
The reason for this is that the cost of an edge $e$ for agent $i$ depends on whether $i$ is currently using it.
To be more precise, if $S = (p_1, \dots, p_k)$ is the current strategy profile and $e \in p_i$, then the cost of $e$ is $c_e(\ell(S))$.
If, however, $e \notin p_i$, then the cost of $e$ is $c_e(\ell(S) + 1)$ as $i$ switching to a strategy that uses $e$ would mean that the load of $e$ increases by $1$.
In other words, the best response aims to minimize $\cost(p_i', (S_{-i}, p_i'))$, i.e., the cost of the edges is not based on $S$ but on the updated profile $(S_{-i}, p_i')$.
To resolve this, we introduce a best-response variant in which the agents ignore the impact of their own change.
In ATA this setting of simultaneous updates based on the costs after the last round is sometimes done under the premise of \emph{capacity restraint} \cite{Hershdorfer1966, sheffi1985urban}.
For a resource selection game,~\citet{Strat_Resour_Selec_Homop_Agent-Gadea23} considered a similar concept and used the terms impact-aware and impact-blind.
Following this, we also call the best-response variant defined above and the resulting Nash equilibrium \emph{impact-aware}.

\subsection{Impact-Blind Best Response}

An \emph{impact-blind best response} of agent $i$ to a strategy profile $S = (p_1, \dots, p_k)$ is a strategy $p_i' \in P_i$ that minimizes $\cost(p_i', S)$.
All definitions for the impact-aware setting can be directly translated (impact-blind Nash equilibrium, sequential/simultaneous, convergence, best-response cycle).
We note that in practical settings, the impact of a single agent changing is usually not very high.
Fortunately, the change from impact-aware to impact-blind has the effect that even simultaneous best response converges.
We note that this is specific to the synergistic setting, which makes intuitive sense for the following reason; see~\Cref{sec:conv-best-resp} for details.
In the synergistic setting, the own impact of changing the strategy only reinforces the decision, which leads to slightly higher stability of the impact-blind variant (the converse is true for the avoidant setting).
Nonetheless, we find it quite surprising that this seemingly minor effect is sufficient to obtain convergence rather than a best-response cycle.

\section{Game Analysis}\label{s:analysis}

Here we analyze the different best-response variants for STA; see~\Cref{tab:variants_matrix} for an overview.
A brief analysis of system optima to complement this chapter can be found in \Cref{a:system_optima}.
We start with the first row (impact-aware) in~\Cref{tab:variants_matrix}.
For the sequential case (first column), convergence already follows from Rosenthal's results~\cite{rosenthal1973class} on congestion games; recall~\Cref{thm:rosenthal-congestion-games}.
For the simultaneous case (second column), we give a best-response cycle in~\Cref{sec:best-response-cycles}.
We additionally give a best-response cycle for a variant that lies between the simultaneous and the sequential setting.
For the second row (impact-blind), we show in~\Cref{sec:conv-best-resp} that for synergistic congestion games (including STA) best response converges to an equilibrium, even in the simultaneous setting.
This notably distinguishes STA from ATA, which has best-response cycles in all cases except for the sequential impact-aware setting covered by~\Cref{thm:rosenthal-congestion-games}.

\begin{table}
  \captionsetup{skip=5pt}
  \centering
  \caption{Overview of our results on the best-response variants of congestion games (CG) and traffic assignment (TA).}
  {
  \setlength\extrarowheight{8pt}
  \begin{tabular}{ p{2.4cm} | p{2.3cm} | p{2.3cm} }
    & \makecell[c]{sequential} & \makecell[c]{simultaneous}
    \\[.5em]
    \hline
    \makecell[c]{impact-aware\\{\footnotesize$\cost(p_i', (S_{-i}, p_i'))$}}
    & \makecell[c]{converges (CG)\\\Cref{thm:rosenthal-congestion-games} \cite{rosenthal1973class}}
    & \makecell[c]{cycles (STA)\\\Cref{lem:proactive_seq_br_cycle}}
    \\[1em]
    \hline
    \makecell[c]{impact-blind\\{\footnotesize$\cost(p_i', S)$}}
    & \multicolumn{2}{c}{\makecell[c]{converges (synergistic CG)\\\Cref{thm:simultaneous_reactive}}}
  \end{tabular}
  }
  \label{tab:variants_matrix}
\end{table}

\subsection{Best-Response Cycles}
\label{sec:best-response-cycles}

Simultaneous best response does not converge for STA, as agents simultaneously attempting to swap towards each other makes it so that they essentially swap past each other and never meet.

\begin{restatable}{observation}{ObsSimIACycle}
  \label{lem:proactive_seq_br_cycle}
  Simultaneous and impact-aware best response for synergistic traffic
  assignment has a best-response cycle.
\end{restatable}

\begin{proof}
  Consider the graph in \Cref{fig:proactive_seq_br_cycle} (left) with the two O-D pairs $(s_1, t_1)$ and $(s_2,t_2)$.
  We call agent $1$ blue and agent $2$ red.
  The cost of the two bold central edges is $1$ for load at most $1$ and $0$ for larger load.
  There are two edges with constant cost $\varepsilon \in (0, 1)$ on the left.
  All remaining edges have cost $0$.
  
  If each agent chooses the cheapest path without assuming the existence of other agents, we get configuration A in \Cref{fig:proactive_seq_br_cycle} (right), i.e., blue chooses the top path and red the bottom path, both with cost $1$.
  If both agents simultaneously play impact-aware best response, then blue swaps to the bottom path and red swaps to the top path, both expecting to reduce their cost to $\varepsilon$.
  This yields configuration B, in which both agents actually have cost $1 + \varepsilon$.
  From there, the best response of both is to swap back, expecting cost $0$.
  However, the simultaneous swap results in configuration A (with cost $1$ for both) from where the cycle repeats indefinitely.
\end{proof}

In the following, we additionally consider a best-response variant in which only agents with the same O-D pair change their strategy simultaneously.
We additionally assume that the agents know of each other, i.e., in their response valuation, they account for the fact that all agents with the same O-D pair will make the same choice.

Formally, we consider two agents to be equivalent if they have the same O-D pair and we call the resulting equivalence classes \emph{groups}.
Then \emph{group-simultaneous best response} refers to a process in which all agents from one group simultaneously change their strategy, while the different groups are considered sequentially.
Consider a step in which one group changes and assume without loss of generality that this group consists of the first $a$ agents, i.e., the group is the set $[a] \subseteq [k]$.
Given a strategy profile $S = (p_1, \dots, p_k)$ and a strategy $p' \in P_1$ that is valid for the group $[a]$, let $(S_{-[a]}, p') = (p', \dots, p', p_{a + 1}, \dots, p_k)$ be the strategy profile where all agents in $[a]$ change their strategy to $p'$.
The \emph{group-impact-aware best response} for agents in $[a]$ is the strategy $p' \in P_1$ for which $\cost(p', (S_{-[a]}, p'))$ is minimized.
We get the following.

\begin{figure}[t]
\captionsetup{skip=5pt}
\centering
\includegraphics{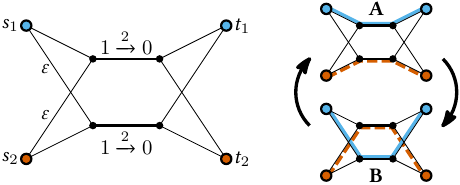}
\caption{Graph (left) with a best-response cycle (right) for simultaneous impact-aware best response.
  The two bold edges have cost $1$ for load below $2$ and cost $0$ for load at least $2$.
  Two edges have cost $\varepsilon$.
  All other edges have cost $0$.}
\label{fig:proactive_seq_br_cycle}
\Description{Graph (left) with a best-response cycle (right) for simultaneous impact-aware best response.
The two bold edges have cost 1 for load below 2 and cost 0 for load at least 2.
Two edges have cost epsilon.
All other edges have cost 0.}
\end{figure}

\begin{restatable}{theorem}{ThmCycleInGroupedSetting}
  \label{thm:cycle-in-grouped-setting}
  Group-simultaneous and group-impact-aware best response for synergistic traffic assignment has a best-response cycle, even if the largest group has size $2$.
\end{restatable}

\begin{proof}
  Consider the graph in \Cref{fig:proactive_seq_mult_brc} (left) with the O-D pairs $(s_i, t_i)$ for $i \in [4]$.
  Note that the agents $1$ and $2$ have the same O-D pair and thus form a group.
  Most edge costs are constant (indicated by the numbers, $0$ if there is no number), except for the three bold edges.
  Here, a cost of $a \oset{\text{$\ell$}}{\rightarrow} b$ indicates that the cost is $a$ for loads below $\ell$ and the cost is decreased to $b$ for load at least $\ell$.
  In accordance with the figure, we refer to the agents as blue ($1$ and $2$), red ($3$), and orange ($4$).
  Also, we refer to the path from $s_1$ to $t_1$ with the three bold edges as \emph{horizontal path}.
  We now first observe that each agent has two relevant options for choosing their path.
  Then we show that the best-response process cycles through the four configurations shown in \Cref{fig:proactive_seq_mult_brc} (right).

  The two options for blue are the top path of cost $20$ (configurations A and D) and the horizontal path with cost depending on red and orange (configuration B and C).
  Moreover, red and orange both have to go up to the horizontal path and then back down.
  Their choice is to either go up early (configurations A and B) or they go up late (configurations C and D).
  They each use their own bold edge on the horizontal path in the early case and share a bold edge in the late case.
  In principle, they could stay longer on the horizontal path, using multiple bold edges, but this would require using the edge of cost $10$ on the horizontal path, which is never good.

  If each agent uses the cheapest path without assuming that any other agents exist, we obtain configuration A in which blue uses the top path (cost $20$) and red and orange go up early.
  As there is no sharing, the cost of red and orange is $5 + 10 = 15$.
  Red and orange cannot improve as each individual change would yield cost $9 + 7 = 16$.
  However, blue can improve.
  As blue is a group of two and we are group-impact-aware, changing blue to the horizontal path reduces the cost of the first two bold edges, yielding cost $1 + 1 + 10 + 7 = 19$ for blue in configuration B.
  Note that this also reduces the cost of red and orange to $1 + 10 = 11$.
  Moreover, it enables an additional improvement for red and orange, as the last bold edge on the horizontal path now has load $2$.
  It is now beneficial for each of them individually to switch to the path going up late, resulting in cost $9 + 1 = 10$.
  If both do the switch (which technically happens sequentially in any order), we end up in configuration C.
  Theses changes of red and orange also impact blue, as its cost is now increased to $5 + 5 + 10 + 1 = 21$.
  With this change, blue has no incentive to stay on the horizontal path and thus goes back to the top path of cost $20$, yielding configuration D.
  This change in turn increases the costs of red and orange to $9 + 7 = 16$ as the last bold edge on the horizontal path has now only load $2$.
  With this, red and orange have no incentive to go up late and thus switch back to their early paths with cost $5 + 10 = 15$ (again, this swap is technically not simultaneous but one after the other).
  This leads us back to configuration A from where the cycling continues.
  \begin{figure*}
    \captionsetup{skip=5pt}
    \centering
    \includegraphics[page=3]{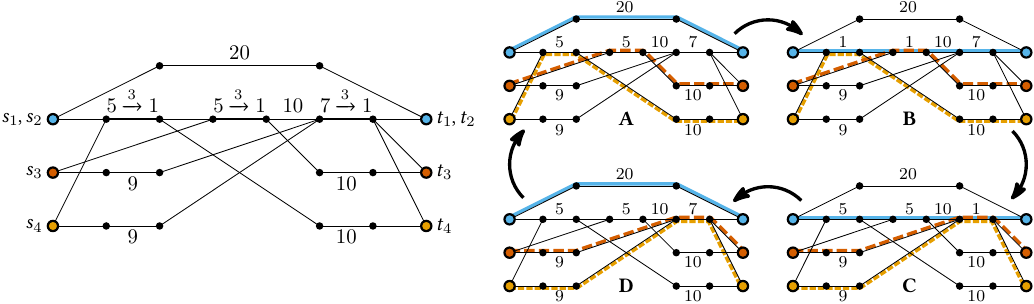}
    \caption{Graph (left) with a best-response cycle (right) for \Cref{thm:cycle-in-grouped-setting}.
      Edges are implicitly oriented from left to right.}
    \label{fig:proactive_seq_mult_brc}
    \Description{Graph (left) with a best-response cycle (right). Edges are implicitly oriented from left to right.}
  \end{figure*}
\end{proof}

Concerning the best-response variants in~\Cref{tab:variants_matrix}, the group-simultaneous and group-impact-aware variant lies between sequential and simultaneous.
Moreover, it lies above impact-aware in the sense that the agents are not only aware of their own impact but also of that of equivalent agents.
We note that we find~\Cref{thm:cycle-in-grouped-setting} somewhat surprising as the grouped setting is quite close to the sequential setting, in which best response converges.
We only need groups of size $2$, which makes the setting almost sequential.
Moreover, using impact-aware but not group-impact-aware best response would also make it equivalent to the sequential setting for the following reason: 
In the impact-aware variant, the group as a whole would swap if and only if the first agent of the group would swap individually.
Moreover, for later agents of the group, the previously swapped agents only reinforce the choice.

\subsection{Convergence of Impact-Blind Best Response}
\label{sec:conv-best-resp}

Here, we show that impact-blind best response converges to an equilibrium for synergistic congestion games.
This is particularly interesting for the simultaneous variant, but our proof also holds for the sequential setting.
Note that the requirement of the congestion game being synergistic is crucial, i.e., the cost of a resource must not increase with a higher load.
For non-synergistic congestion games like ATA the impact blindness is in fact detrimental.
Thus, their impact-blind formulations can run into best-response cycles.

Our proof uses a potential function $\phi \colon P_1 \times \dots \times P_k \to \mathbb R$ that maps a strategy profile $S$ to a number.
Assume agents change their strategy using impact-blind best response, yielding a new strategy profile $S'$.
We show that in this case $\phi(S') < \phi(S)$ holds.
As the potential cannot decrease indefinitely (the strategy space is finite), the process must reach an equilibrium.

We use the potential function introduced by Rosenthal~\cite{rosenthal1973class} for the sequential and impact-aware setting.
It sums for each resource the cost of all loads up to the actual load.
To make this more formal, recall that $c_e(\ell)$ is the cost of resource $e$ for load $\ell$ and that $\ell_e(S)$ refers to the load of $e$ given the strategy profile $S$.
The \emph{edge potential} $\phi_e(S)$ of $e$ is defined as $\phi_e(S) = \sum_{\ell = 0}^{\ell_e(S)} c_e(\ell)$.
The \emph{potential} $\phi$ is the sum over all resources, i.e., $\phi(S) = \sum_{e \in E}\phi_e(S)$.

Our proof that $\phi$ decreases with each round of impact-blind best response works roughly as follows.
Given strategy profile $S = (p_1, \dots, p_k)$, each agent $i$ changes their strategy from $p_i$ to $p_i'$ such that $\cost(p_i', S)$ is minimized.
With this change, agent $i$ anticipates that their cost decreases by $\Delta_i = \cost(p_i, S) - \cost(p_i', S)$.
Note that the actual cost decrease can be different as the agent did not factor in their own impact, or even the impact of the other agents changing at the same time in the simultaneous setting.
However, we can show that the potential decreases at least by the total anticipated cost of all agents.
As $\Delta_i \ge 0$ (otherwise the agent would not have changed strategy), this results in a potential decrease.

\begin{theorem}\label{thm:simultaneous_reactive}
  In a synergistic congestion game, impact-blind best response (simultaneous or sequential) converges to an equilibrium.
\end{theorem}

\begin{proof}
  We consider the general setting where each round of best response consists of any non-empty subset of agents changing their strategy.
  This in particular includes the simultaneous and sequential variants.
  Let $S = (p_1, \dots, p_k)$ be the current strategy profile and let $S' = (p_1', \dots, p_k')$ be the resulting profile after the change.
  Let $\Delta_i$ be the anticipated cost decrease for agent $i$ as defined above and let $\Delta = \sum_{i = 1}^k \Delta_i$.
  As $\Delta_i > 0$ for some $i$, we also have $\Delta > 0$.

  We now change perspective on the anticipated cost decrease from the agents to the resources.
  For $e \in p_i \setminus p_i'$, agent $i$ changed from using $e$ to not using $e$.
  Thus, the resource $e$ contributes $c_e(\ell_e(S))$ positively to $\Delta_i$.
  Conversely, if $e \in p_i' \setminus p_i$, then $e$ contributes $-c_e(\ell_e(S))$ to $\Delta_i$.
  To define the overall contribution of $e$ to $\Delta$, note that the load of $e$ decreased by $\ell_e(S) - \ell_e(S')$.
  Each of these load decreases corresponds to an agent anticipating a cost decrease of $c_e(\ell_e(S))$, thus we define the anticipated cost decrease of an edge $e$ as $\Delta_e = c_e(\ell_e(S)) \cdot (\ell_e(S) - \ell_e(S'))$.
  It directly follows that $\sum_{e \in E} \Delta_e = \Delta > 0$.

  Recall that our goal is to show that $\phi(S) - \phi(S') > 0$.
  For this, we show that for every resource $e \in E$ it holds that $\phi_e(S) - \phi_e(S') \ge \Delta_e$, which implies $\phi(S) - \phi(S') \ge \Delta > 0$.
  To show $\phi_e(S) - \phi_e(S') \ge \Delta_e$, note that the potential difference is
  \begin{equation*}
    \phi_e(S) - \phi_e(S') = \sum_{\ell = 0}^{\ell_e(S)} c_e(\ell) - \sum_{\ell = 0}^{\ell_e(S')} c_e(\ell).
  \end{equation*}
  First consider the case that $\ell_e(S) \ge \ell_e(S')$, i.e., decreasing load.
  Then we get the following, where the inequality holds as the sum has $\ell_e(S) - \ell_e(S')$ terms and $c_e$ is non-increasing
  \begin{equation*}
    \phi_e(S) - \phi_e(S') = \sum_{\ell = \ell_e(S') + 1}^{\ell_e(S)} c_e(\ell) \ge
    (\ell_e(S) - \ell_e(S')) \cdot c_e(\ell_e(S)) = \Delta_e.
  \end{equation*}
  For the case $\ell_e(S) \le \ell_e(S')$, we analogously obtain
  \begin{equation*}
    \phi_e(S) - \phi_e(S') = - \sum_{\ell = \ell_e(S) + 1}^{\ell_e(S')} c_e(\ell) \ge
    (\ell_e(S) - \ell_e(S')) \cdot c_e(\ell_e(S)) = \Delta_e,
  \end{equation*}
  which concludes the proof.
\end{proof}

As STA is a synergistic congestion game, the above theorem holds for STA.
In the remainder of the paper, we utilize the fact that the simultaneous variant of impact-blind best response converges.

\section{STA Algorithm}
\label{s:staAlgo}

In a traffic assignment game with converging best response, computing an equilibrium essentially boils down to iterated shortest-path computations.
To make this efficient, one can potentially apply speed-up techniques that have been developed in the context of route planning~\cite{DBLP:series/lncs/BastDGMPSWW16}.
These techniques exploit the fact that the graph does not change by
pre-computing auxiliary information once in the beginning, which is then used to enable highly efficient queries.

For traffic assignment, the graph topology does not change between shortest-path queries, but the edge costs change depending on the load.
Fortunately, there are speed-up techniques that support such cost changes~\cite{DBLP:journals/jea/DibbeltSW16,DBLP:journals/transci/DellingGPW17}, which work in three phases.
The first phase is a so-called \emph{metric-independent preprocessing}, where pre-computations are done based only on the topology of the graph, ignoring the edge costs.
The second phase is called \emph{customization} and performs pre-computations based on the edge costs.
The third phase are the shortest-path \emph{queries}, which are highly efficient due to the data structures computed in the first two phases.

Concerning the different variants for STA, impact-blind and simultaneous best response is particularly suited for this algorithmic approach as all agents do shortest-path computations with respect to the same cost function once before updating it due to the changed load.
This is the reason why we focus on impact-blind and simultaneous best response for computing STA equilibria.
We note that this is viable as this best-response variant converges for STA due to \Cref{thm:simultaneous_reactive}, which is not true for ATA.

\parheader{Algorithm description}

With these preliminary considerations, we are now ready to state our algorithm.
We have a \emph{load vector} storing the load of each edge in the current traffic assignment.
It is initialized with all $0$s.
Moreover, we do the metric independent preprocessing once in the beginning.
Then, in each iteration, we first do the customization based on the edge costs defined by the current load vector.
Afterwards, the resulting data structure is used to compute a shortest path for every O-D pair.
Finally, the load vector is updated according to the resulting paths.
This is iterated until no agent changes their path, which is equivalent to checking whether the load vector changed, due to the proof of \Cref{thm:simultaneous_reactive}.

\parheader{Implementation details}

There are different variants of the above-mentioned three-phase approach with different trade-offs between the running times of the different phases.
In our setting, the bottleneck is the query.
We thus use a so-called \emph{customizable contraction hierarchy (CCH)}~\cite{DBLP:journals/jea/DibbeltSW16}, which provides faster queries at the cost of slower customization compared to CRP~\cite{DBLP:journals/transci/DellingGPW17}.
Explaining CCH is beyond the scope of this paper and we refer the reader to the recent survey~\cite{Blaesius2025}.
For readers familiar with CCH, we want to briefly mention the details of our implementation.
We compute the nested dissection order with InertialFlowCutter~\cite{gottesbueren2019faster}.
We use perfect customization, elimination tree queries, and all improvements stated by \citet{Blaesius2025}.
We do not use parallelization or SIMD operations.

    \section{Experimental Evaluation}
    \label{sec:experimental_evaluation}
    We provide experimental results on the convergence and resulting traffic flows of our STA algorithm using real-world data and an implementation in C++17%
    \footnote{The code is available at \url{https://github.com/molaupi/synergistic-traffic-assignment}.}.
    We compile with \texttt{gcc} $13.3.1$ using \texttt{-O3}.
    We use a machine with Fedora $39$ (kernel $6.8.11$), $\SI{32}{GiB}$ of DDR$4$-$4266$ RAM, and an AMD Ryzen $7$ PRO $5850$U CPU clocked at $\SI{4.40}{Ghz}$.
    
    \subsection{Inputs and Methodology}
    \label{subsec:experiments_inputs_and_methodology}
    We evaluate our algorithm on the road network of the city of Stuttgart, Germany, and the surrounding region.
    We use a road network based on the OSM network of Germany%
    \footnote{Available at \url{https://download.geofabrik.de/europe/germany.html}.}
    and the OSM relation of the Stuttgart city limits%
    \footnote{See \url{https://www.openstreetmap.org/relation/2793104}.}.
    We only consider travel demand within these city limits but allow travelers to also use major streets in the surrounding area.
    For this, we include roads of OSM rank at least \emph{tertiary link}%
    \footnote{See \url{https://wiki.openstreetmap.org/wiki/Key:highway}.}
    in a rectangular box around the city that is $3$ times the width and height of the bounding box of the city limits.
    The resulting graph contains $\num{70915}$ vertices and $\num{138862}$ edges.
    We retrieve traversal times $d(e)$ for each edge $e$ using the length and speed limit of each road segment encoded in the OSM data.
    
    Our demand data~\cite{schlaich2011verkehrsmodellierung} models travel demand in Stuttgart for a whole week and was originally forecast using mobiTopp \cite{mallig2013mobitopp,mallig2015modeling}, which was calibrated from a household travel survey conducted in 2009/2010~\cite{stuttgart2011mobilitat}.
    We use the same test scenarios that \citet{buchhold2019real} use to evaluate ATA algorithms, i.e., a typical morning peak, a typical evening peak, a typical day and a whole week; see~\Cref{tab:scenarios}.
    
    \begin{table}[t]
      \captionsetup{skip=5pt}
      \centering
      \caption{
        Four traffic Scenarios used for the evaluation of STA.
      }\label{tab:scenarios}
      \begin{tabularx}{\linewidth}{LLR}
        \toprule
        Scenario       & Analysis period     & O-D pairs \\
        \midrule
        \texttt{S-morn} & Tue., 7:30-8:30am   &     $\num{32034}$ \\
        \texttt{S-even} & Tue., 4:30-5:30pm &     $\num{38021}$ \\
        \texttt{S-day}  & a whole Tuesday     &    $\num{453926}$ \\
        \texttt{S-week} & a whole week        & $\num{2946810}$ \\
        \bottomrule
      \end{tabularx}
    \end{table}
    
    The cost function for STA may significantly depend on the particular application of STA, e.g., the type of vehicle used.
    Here we use a simple function with a single tuning parameter $r\in[0,1]$, which we call the \emph{selfishness parameter}.
    This allows us to interpolate between ignoring sharing ($r=1$) and aggressively optimistic sharing ($r=0$).
    Let $d(e)$ and $\ell_e$ be the travel time (assuming free flow) and the load of edge $e$.
    We define the cost $c_e^r(\ell)$ of $e$ as
    \begin{equation*}
      c_e^r(\ell_e) = r\cdot{}d(e) + (1-r)\cdot\frac{d(e)}{\ell_e + 1}~.
    \end{equation*}

    \subsection{Sharing Evaluation}
    \label{subsec:experiments_evaluation}
    
    Due to our simple cost function, we evaluate sharing using quantities that are independent of the concrete definition of $c_e^r$.
    Let $\distance(u,v)$ be the length of the shortest path between $u$ and $v$ with respect to the travel time $d$.
    Further, we define $D$ as a natural extension of $d$ on sets of edges, i.e., $D(p) = \sum_{e\in p}d(e)$.
    Recall that $p_i$ denotes the path of agent $i \in [k]$ and  $\ell_e$ is the load on edge~$e$.

    We define the \emph{average stretch} as
    \begin{equation*}
      \frac{1}{k}\sum_{i\in [k]} \frac{D(p_i)}{\distance(s,t)}~.
    \end{equation*}

    The \emph{average sharing} is
    \begin{equation*}
      \frac{1}{k}\sum_{i\in [k]}\frac{\sum_{e\in p_i}d(e)\cdot(\ell_e-1)}{D(p_i)}~,
    \end{equation*}
    i.e., the sharing of each agent $i \in [k]$ is the number of other agents they share their path with, averaged over the whole path.
    In other words, if you pick a random point during $i$'s journey, then $i$'s sharing value is the expected number of other agents on $i$'s current edge.
    As there is also sharing without STA (i.e., for $r = 1$), we define the \emph{normalized average sharing} as the average sharing relative to the average sharing obtained for $r = 1$.
    
    \Cref{fig:stretch_sharing} shows the trade-off between average stretch and normalized average sharing for different choices of $r$.
    Stretch and sharing both increase with lower $r$, seemingly in a linear relation.
    However, the slope is substantially below $1$, indicating that slightly longer paths can result in higher sharing.
    For instance, at $r\approx 0.0075$, the sharing can be doubled at the cost of only $1.25$ times longer paths.

    \Cref{fig:sharing_potential} analyzes the distribution of sharing potential along riders' paths.
    The plot has one curve per pair of selfishness $r \in \{0, 1\}$ and load threshold $\ell \in \{1, 10, 100\}$.
    Each curve shows the fraction of agents (y-axis) for whom at least a fraction $x$ of their path (x-axis) is shared with at least $\ell$ others in the $r$ STA flow.
    The curves for STA with strong sharing ($r = 0$) stay high even for large values of $\ell$, particularly in comparison to just using individual shortest paths ($r = 1$).
    Thus, most agents achieve high sharing on a large portion of their path, which is important for applications like ride sharing.
    
    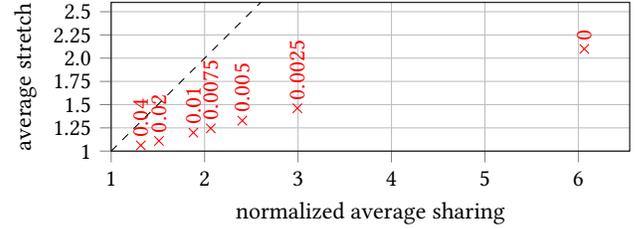
\begin{figure}[t]
      \captionsetup{skip=5pt}
      \centering
      \begin{tikzpicture}
        \begin{axis}[
          every y tick label/.append style={inner ysep=0pt, outer ysep=0pt},
          width=\columnwidth,
          height=0.42\columnwidth,
          ymin=1,
          ymax=2.6,
          ytick={1,1.25,...,3},
          yticklabels={1,1.25,1.5,1.75,2.0,2.25,2.5},
          ylabel={\makebox[0pt]{average stretch}},
          xmin=1,
          xmax=6.55,
          xlabel={normalized average sharing},
          xtick={1,2,3,4,5,6},
          xticklabels={1,2,3,4,5,6},
          grid=major,
          ytick align=outside,
          xtick align=outside,
          ytick pos=left,
          xtick pos=left,
          ]
          \newcommand{\addPoint}[5][rotate=90,anchor=west]{
            \addplot[mark=x,red,mark size=2.5pt] coordinates {(#5,#3)} node[#1] {#2};
          }
          \addPoint{0}{2.099742}{1491.663019}{6.06420911843301}
          \addPoint{0.0025}{1.460697}{736.277759}{2.9932647273246}
          \addPoint{0.005}{1.327979}{591.241643}{2.40363467955487}
          \addPoint{0.0075}{1.244435}{508.271926}{2.06632946518915}
          \addPoint{0.01}{1.198864}{462.449586}{1.88004325408349}
          \addPoint{0.02}{1.109245}{371.948767}{1.51212108613072}
          \addPoint{0.04}{1.061457}{324.177046}{1.31790985852675}
          \addplot[dashed,domain=1:6.6] {x};
        \end{axis}
      \end{tikzpicture}
      \caption{The normalized average sharing in relation to the average stretch in the morning scenario for different values of $r$; the numbers at the points indicate $r$.}\label{fig:stretch_sharing}
      \Description{The normalized average sharing in relation to the average stretch in the morning scenario for different values of r; the numbers at the points indicate r.}
    \end{figure}

    \definecolor{sta100}{HTML}{2171b5}
    \definecolor{sta10}{HTML}{6db1cc}
    \definecolor{sta1}{HTML}{c7dcf0}
    \definecolor{ff100}{HTML}{d94801}
    \definecolor{ff10}{HTML}{fd8d3c}
    \definecolor{ff1}{HTML}{fdd0a3}

    \begin{figure}[t]
      \captionsetup{skip=5pt}
      \centering
      \begin{tikzpicture}
        \begin{axis}[
          every axis y label/.append style={inner xsep=0pt, outer xsep=0pt},
          width=\columnwidth-0.55cm,
          height=5.8cm,
          xmin=0,
          xmax=1,
          xlabel={fraction $x$ of each agent's path},
          xtick={0,0.1,...,1},
          ymin=0,
          ymax=1.05,
          ylabel={fraction of O-D pairs that share\\$\ge x$ of their path with $\ge \ell$ others},
          ylabel style={align=center},
          ytick={0,0.1,...,1},
          legend cell align={left},
          legend style={at={(0.03,0.03)},anchor=south west},
          grid=major,
          xtick align=outside,
          ytick align=outside,
          xtick pos=left,
          ytick pos=left,
          ]
          \addplot[sta1, thick] table [x=x, y=y, col sep=comma] {sharingk1.csv};
          \addlegendentry{$r = 0$, $\ell=1$}
          \addplot[sta10, thick] table [x=x, y=y, col sep=comma] {sharingk10.csv};
          \addlegendentry{$r = 0$, $\ell=10$}
          \addplot[sta100, thick] table [x=x, y=y, col sep=comma] {sharingk100.csv};
          \addlegendentry{$r = 0$, $\ell=100$}
          
          \addplot[ff1, thick, dashed] table [x=x, y=y, col sep=comma] {ffsharingk1.csv};
          \addlegendentry{$r = 1$, $\ell=1$}
          \addplot[ff10, thick, dashed] table [x=x, y=y, col sep=comma] {ffsharingk10.csv};
          \addlegendentry{$r = 1$, $\ell=10$}
          \addplot[ff100, thick, dashed] table [x=x, y=y, col sep=comma] {ffsharingk100.csv};
          \addlegendentry{$r = 1$, $\ell=100$}
        \end{axis}
      \end{tikzpicture}
      \caption{
        Fraction of agents that share a fraction of at least $x$ of their travel time with at least $\ell$ other agents.  We use $\ell \in \{1, 10, 100\}$ and show values for $r = 0$ (high sharing) and $r = 1$ (free flow) in the \texttt{S-morn} scenario.
      }\label{fig:sharing_potential}
      \Description{Fraction of agents that share a fraction of at least x of their travel time with at least l other agents.  We use l in \{1, 10, 100\} and show values for r = 0 (high sharing) and r = 1 (free flow) in the \texttt{S-morn} scenario.}
    \end{figure}
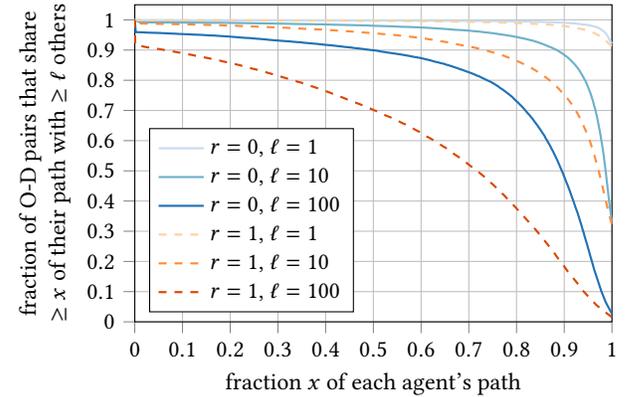

    \subsection{Performance Evaluation}
    \label{sec:perf-eval}
    Our theoretical results prove convergence, but give no guarantees for how many iterations are necessary.
    Thus, in this section, we evaluate the number of iterations and running times in practice.

    \Cref{fig:iterations} shows the number of iterations depending on $r$.
    Smaller values of $r$ require more iterations but even for $r=0$, fewer than \num{20} iterations suffice.
    The number of agents appears to have no effect.
    
    To give an impression of the overall running times, the entire process with $r=0$ takes about one second for \texttt{S-morn} and \texttt{S-even}, ten seconds for \texttt{S-day}, and a minute for \texttt{S-week}, respectively. 

    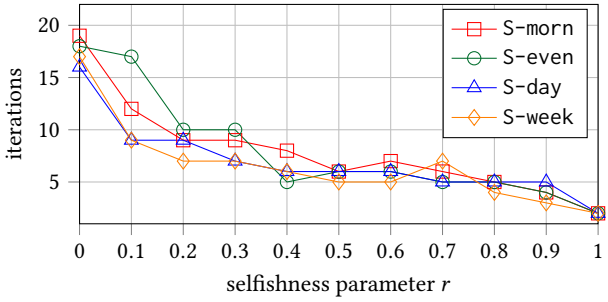
\begin{figure}[t]
      \captionsetup{skip=5pt}
      \centering
      \begin{tikzpicture}
        \definecolor{green}{RGB}{0,109,44}
        \begin{axis}[
          width=\columnwidth,
          height=4.5cm,
          xmin=0,
          xmax=1,
          xlabel={selfishness parameter $r$},
          ymin=1,       
          ymax=22,
          ylabel={iterations},
          xtick={0,0.1,...,1},
          legend cell align={left},
          grid=major,
          xtick align=outside,
          ytick align=outside,
          xtick pos=left,
          ytick pos=left,
          ]
          \addplot[mark=square,red,mark size=2.5pt] table [x=r, y={morn}, col sep=comma] {iterations.csv};
          \addlegendentry{\texttt{S-morn}}
          \addplot[mark=o,green,mark size=2.5pt] table [x=r, y={even}, col sep=comma] {iterations.csv};
          \addlegendentry{\texttt{S-even}}
          \addplot[mark=triangle,blue,mark size=3pt] table [x=r, y={day}, col sep=comma] {iterations.csv};
          \addlegendentry{\texttt{S-day}}
          \addplot[mark=diamond,orange,mark size=3pt] table [x=r, y={week}, col sep=comma] {iterations.csv};
          \addlegendentry{\texttt{S-week}}
        \end{axis}
      \end{tikzpicture}
      \caption{The number of iterations for the STA best response to converge.
          The reported numbers include the last iteration that did not result in a change.}\label{fig:iterations}
      \Description{The number of iterations for the STA best response to converge. The reported numbers include the last iteration that did not result in a change.}
    \end{figure}
      
    In \Cref{fig:performance}, we can see that in later iterations, the shortest-path queries become noticeable faster.
    For $r=0$, the time spent on queries almost halves after only $6$ iterations.
    As the number of O-D pairs stays the same, this implies that the CCH (with perfect customization) can leverage the increased sharing and smaller number of relevant edges for faster queries.

    \begin{figure}[t]
      \captionsetup{skip=5pt}
      \centering
      \begin{tikzpicture}
        \begin{axis}[
          width=\columnwidth,
          height=4.5cm,
          xmin=0,
          xmax=37.5,
          xlabel={iteration\strut},
          ymin=0,
          ylabel={time [ms]},
          legend cell align={left},
          grid=major,
          xtick align=outside,
          ytick align=outside,
          xtick pos=left,
          ytick pos=left,
          ]
          \addplot[mark=square,red!25!blue,mark size=2.5pt] table [x=iteration, y={Time0.1}, col sep=comma] {iterations_time.csv};
          \addlegendentry{$r=0.1$}
          \addplot[mark=o,red!50!blue,mark size=2.5pt] table [x=iteration, y={Time0.01}, col sep=comma] {iterations_time.csv};
          \addlegendentry{$r=0.01$}
          \addplot[mark=triangle,red!75!blue,mark size=3pt] table [x=iteration, y={Time0.001}, col sep=comma] {iterations_time.csv};
          \addlegendentry{$r=0.001$}
          \addplot[mark=diamond,red,mark size=3pt] table [x=iteration, y={Time0.0}, col sep=comma] {iterations_time.csv};
          \addlegendentry{$r=0$}
        \end{axis}
      \end{tikzpicture}
      \caption{ Time spent for shortest-path computations and load accumulation per iteration in the morning scenario.
        }\label{fig:performance}
      \Description{Time spent for shortest-path computations and load accumulation per iteration in the morning scenario.}
    \end{figure}
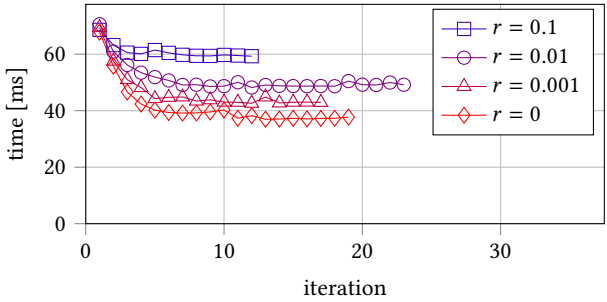

    As claimed previously, we observe that shortest-path computations are indeed the bottleneck compared to customization.
    For instance, for $r = 0$, the smallest scenario (\texttt{S-morn}), and the fastest iteration, the time spent on shortest-path computations is roughly $\SI{37}{ms}$.
    Meanwhile, an average of only $\SI{10.5}{ms}$ is spent on customization per iteration.
    It thus makes sense to opt for a CCH variant with more efficient queries at the cost of slower customization.
    In the future, shortest-path queries can be optimized further, e.g., by employing thread- and instruction-level parallelism (cf.~\cite{buchhold2019real}).
    
    \subsection{Bus Line Planning}
    \label{subsec:experiments_buses}
    As a proof of concept for the merits of STA, we use STA in a simple bus line planning problem.
    We propose a greedy path-based line planning algorithm to find trunk bus lines for a feeder-trunk system.
    We show that bus lines based on STA paths reduce total vehicle travel time compared to lines based on free-flow (FF) paths. 
    Note that this proof-of-concept model is not optimized for realism but for simplicity while grasping the effects of sharing between travelers.

    \parheader{Bus Line Problem}
    We consider a line planning problem where we are given a road network and travel demands.
    The goal is to find bus lines subject to a bus operation time budget such that the total vehicle travel time (for buses and feeder vehicles) is minimized.
    We assume a road network $G=(V,E)$ with travel times $d(e)$ for every $e \in E$, a set of O-D pairs $X=\{(s_x, t_x)\}$, a bus operation time budget $B$, a fixed service frequency $f$ for every line, and an operation time window $[0,T]$ s.t. every journey fits within that window.

    We consider a bus line $L$ to be a simple path in $G$.
    Buses travel the length of $L$ from start to finish.
    Let $\lineLength{L}$ be the travel time from start to finish. 
    The total vehicle time needed to service $L$ during $[0,T]$ with frequency $f$ is $\lineVehTT{L} = T \cdot f \cdot \lineLength{L}$.

    Travelers may enter and exit a bus anywhere on the line.
    For any part of a journey that is not covered by bus lines, the traveler has to use a feeder (e.g., a taxi).
    A feeder is immediately available between any two places but it can never be shared with other travelers.

    Our goal is to minimize \emph{total vehicle operation time (TVOT)}, i.e., the sum of operation times of all buses and feeder vehicles.
    Without any buses, every traveler uses the feeder system to go directly from their origin to their destination without any sharing.
    This gives us a baseline TVOT that we then aim to reduce by introducing buses.
    More specifically, we want to find a set of bus lines $\mathcal{L}$ such that the TVOT is minimized but the sum of bus operation times fits into the bus operation time budget $B$, i.e., $\sum_{L \in \mathcal{L}} \lineVehTT{L} \le B$.  

    \parheader{Algorithm}
    We use a simple path-based bus line planning algorithm to compare lines based on STA paths to lines based on FF paths. 
    As input, the algorithm takes a set of paths $P$, containing a path $p_x$ from $s_x$ to $t_x$ for every O-D pair $x \in X$.
    Let $N_P(e)$ be the number of paths in $P$ that contain $e$.
    We greedily construct lines:

    A new line $L$ is started at the edge $e \in E$ with the largest $N_P(e)$.
    We extend $L$ forwards by iteratively appending edges at the end, always choosing the edge $e'$ with the largest $N_P(e')$.
    We stop extending $L$ forwards if there is no next edge with $N_P(e') > 0$.
    We then extend the line backward analogously.

    Then, for each path $p \in P$ that overlaps with $L$, we remove $p$ from $P$, slice the part covered by $L$ out of $p$, and add the remaining subpaths of $p$ back to $P$.
    We update $N_P(e)$ for $e \in E$ accordingly.
    Note that we do this for only a subset of paths such that a fixed seating capacity of buses is not exceeded anywhere on the line.

    We keep generating new lines until all traveler demand is covered by bus lines, i.e., until $P$ becomes empty.
    We use a knapsack solver to choose a subset of lines $\mathcal{L}$ that fits into the operation time budget $B$ and maximizes traveler coverage. 

    \begin{figure}[t!b]
        \captionsetup{skip=5pt}
        \centering
        \begin{tikzpicture}
            \begin{axis}[
            	every x tick label/.append style={text depth=0pt},
                width=\columnwidth-0.5cm,
                height=4.5cm,
                xmin=0,
                xlabel={bus budget $B$ [\si{h}]\strut},
                ymin=0,
                ytick={0, 2500, 5000, 7500, 10000},
                scaled ticks=false,
                ylabel={TVOT [\si{h}]},
                legend cell align={left},
                grid=major,
                xtick align=outside,
                ytick align=outside,
                xtick pos=left,
                ytick pos=left,
                ]
                \addplot[mark=pentagon,blue,mark size=3pt] table [x=budget, y={FF}, col sep=comma] {bus_eval.csv};
                \addlegendentry{$r=1$}
                \addplot[mark=o,red!33!blue,mark size=2.5pt] table [x=budget, y={0.01}, col sep=comma] {bus_eval.csv};
                \addlegendentry{$r=0.01$}
                \addplot[mark=diamond,red,mark size=3pt] table [x=budget, y={0}, col sep=comma] {bus_eval.csv};
                \addlegendentry{$r=0$}
            \end{axis}
        \end{tikzpicture}
        \caption{Total vehicle operation time (in hours) for varying bus budgets and $r \in \{0, 0.01, 1\}$.
        }
        \label{fig:bus_eval_ratio_budget_ridertt}
        \Description{Total vehicle operation time (in hours) for varying bus budgets and r in \{0, 0.01, 1\}}
    \end{figure}
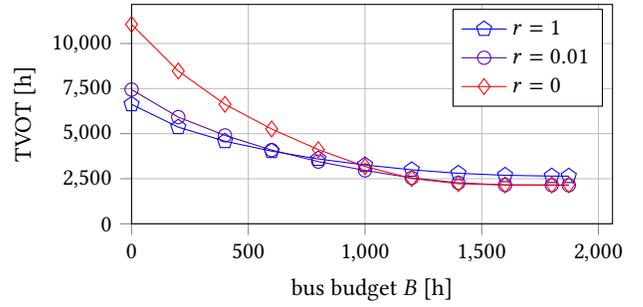
    \parheader{Evaluation}
    We evaluate bus lines for the \texttt{S-morn} instance (cf.~\Cref{tab:scenarios}).
    We assume buses with a capacity of $80$ and a fixed service frequency of $f = 1/10\si{\minute}$.
    For the input paths $P$, we use STA paths with different values of the selfishness parameter $r$ (see~\Cref{subsec:experiments_inputs_and_methodology}), with $r=1$ being equivalent to using FF paths. 
    For simplicity's sake, we restrict each traveler to their input path and allow them to only use bus lines with which they overlap during line construction.

    In~\Cref{fig:bus_eval_ratio_budget_ridertt}, we show the TVOT for different bus budgets and different selfishness $r \in \{0, 0.01, 1\}$.
    Note that with a budget of $B=\SI{0}{h}$ and selfishness of $r=1$, travelers use only feeder vehicles on shortest paths, giving us the baseline TVOT of about $\SI{6600}{h}$.

    Unsurprisingly, bus lines reduce the TVOT for any value of $r$.
    Thus, we consider the minimal TVOT for different values of $r$. 

    As previously stated in~\Cref{subsec:experiments_evaluation}, STA paths are longer with smaller selfishness.
    Since we restrict travelers to their input paths even when using feeder vehicles, these longer paths result in a larger TVOT with smaller selfishness $r$ if the bus budget is small.
    This effect is particularly strong for $r=0$, which is why we focus on the comparison between $r=0.01$ and $r=1$ (FF) in the following.

    Given a large enough budget, the improved sharing with STA leads to better bus lines which make up for the longer paths and reduce the TVOT to smaller values than FF-based lines.
    At a bus budget of $B=\SI{1800}{h}$, STA-based lines with $r=0.01$ lead to a minimum TVOT of $\SI{2140}{h}$ while FF-based lines require at least $\SI{2641}{h}$.
    Thus, the FF-based lines require more than twice as much additional feeder service ($\SI{841}{h}$) compared to the STA-based lines ($\SI{340}{h}$). 
    This reduction of about $\SI{500}{h}$ is equivalent to $7.5\%$ of the baseline TVOT for using only feeder vehicles.

    \section{Conclusions and Future Work}
    We have demonstrated that synergistic traffic assignment is an interesting game-theoretic concept that allows fast computations of equilibria. 
    These equilibria seem relevant for applications in the context of envisioned software-defined transportation systems. 
    A lot remains to be done there. 
    Cost functions could try to better model costs and behavior of agents (perhaps using simulation or machine learning). 
    More sophisticated algorithms to interpret the equilibria would take into account actually available vehicles and constraints on using them. 
    Perhaps more interestingly, agents could use equilibrium information (based on historical data) to look for routes and vehicles worth sharing. 
    This brings us back to the interesting algorithmic question of how STA equilibria can be computed/approximated in a dynamic fashion.

\begin{acks}
	This work was supported by funding from the pilot program Core Informatics of the Helmholtz Association (HGF).
	This paper was created in the “Country 2 City - Bridge” project of the „German Center for Future Mobility“, which is funded by the German Federal Ministry for Digital and Transport.
\end{acks}

\bibliographystyle{ACM-Reference-Format} 
\bibliography{references}


\appendix

\section{Appendix: System Optima}\label{a:system_optima}
  Since sharing potential in the STA setting may be utilized by a central institution, e.g. by providing a public transit system, there may also be merit in considering system optima.
  Unlike equilibria, which consider the situation from the viewpoint of passengers, seeking to optimize their personal gain from sharing while not increasing the trip's duration too much, system optima describe the situation from the perspective of a central institution seeking to maximize sharing even at the expense of some individuals to maximize the overall utility of sharing.
  The following results hold independent of the specific update variants discussed in \Cref{tab:variants_matrix}.
  
  First, note that both price of anarchy and price of stability for STA are unbounded, since a system optimum may require all agents to use the same resource to massively reduce its cost, but a single agent may always diverge from that strategy in favor of a slightly cheaper resource.
  As such, individual passengers may increase the system's cost arbitrarily with their selfish behavior.
  Note that this is the same as for ATA, where the price of stability is unbounded for similar reasons.
  However, it is unclear whether these selfish passengers are even accounted for realistically in a system optimum, since they may just choose to not participate in a shared system that they themselves do not profit sufficiently from.
  
  A further advantage of working with equilibria over system optima is that we can efficiently compute them, while finding a system optimum for STA is NP-hard.
  \begin{theorem}
    \label{trm:opt_STA_hard}
    Finding a system optimum for STA is NP hard.
  \end{theorem}
  \begin{proof}
    We reduce from a SAT variant, where each literal is used in exactly two different clauses.
    A SAT instance has clauses $C=\{C_1,\dots,C_m\}$ and variables $V=\{v_1,\dots,v_n\}$.
    For every clause $C_i$, we create one edge $c_i$ and one O-D pair that has to use only that edge.
    Each such edge should have cost $3n$ for less then $2$ flow and zero cost for at least $2$ flow.
    For every variable $v_i$, we create one O-D pair $x_i$ with unique origin $s_i$ and destination $t_i$.
    It can use one of two paths, one for each literal $v_i$ and $\lnot v_i$, constructed as follows;
    Let the two clauses that use our literal be $C_i,C_j$.
    Then, the literal's path moves from $s_i$ to $c_i$ to $c_j$ and finally to $t_i$.
    We assign each of the three edges in this path not associated to a clause cost $1$ independent of the flow.
    If one clause $C_i$ or $C_j$ uses both literals $v_i$ and $\lnot v_i$, ensure that either both literal's paths use it first or both use it second to avoid unintentional shortcuts.
    Now each clause $C_i$ has to be satisfied to avoid the high cost of $3n$ caused by $c_i$ and its corresponding O-D pair.
    Each variable can only choose one literal to be true and that literal has to follow the intended path over exactly its two clauses to maintain the minimum possible cost of $3$ induced by connecting edges.
    Thus, asking whether a synergistic traffic assignment with cost at most $3n$ exists is equivalent to asking whether the given SAT instance is satisfiable.
  \end{proof}

\end{document}